\documentclass[conference]{IEEEtran}
\usepackage{fancyhdr}
\usepackage{amsmath,amssymb}
\usepackage{amsfonts}
\usepackage[dvips]{graphicx}
\usepackage{dsfont}

\usepackage{amsmath}
\usepackage{color}
\usepackage{pgfplots}           
\pgfplotsset{compat=1.16}
\usepackage{enumerate}
\usepackage{graphicx}
\usepackage{cite}
\usepackage{mathrsfs}
\usepackage{subcaption}
\usepackage{stfloats}
\usepackage{xfrac}
\usepackage{verbatim}
\usepackage{tikz}
\usepgfplotslibrary{groupplots}
\usepackage{mathdots}
\usepackage{yhmath}
\usepackage{cancel}
\usepackage{color}
\usepackage{siunitx}
\usepackage{array}
\usepackage{multirow}
\usepackage{textcomp}
\usepackage{gensymb}
\usepackage{tabularx}
\usepackage{booktabs}
\usepackage{graphicx}
\usepackage{xcolor}
\usepackage{makecell}
\usepackage{subcaption}
\usepackage[caption=false,font=footnotesize]{subfig}
\usepackage{svg}
\usepackage{epstopdf}
\usepackage{graphicx}
\usepackage{bm}
\usepackage{booktabs}
\usepackage{siunitx}
\usepackage{mathtools,amssymb}
\usepackage[noend]{algpseudocode}
\usepackage{algorithm}
\usepackage{varwidth}  
\allowdisplaybreaks
\newtheorem{theorem}{Theorem}

\newtheorem{remark}{Remark}
\newtheorem{corollary}{Corollary}
\DeclareMathOperator{\sinc}{sinc}
\usepackage{soul}
\usepackage{url}
\usepackage{balance}
 
\allowdisplaybreaks


\begin{document}
\IEEEoverridecommandlockouts
	
\title{Closed-form Model for Radiation Pattern of Pinching Antennas}

\author{\IEEEauthorblockN{Muhammad Zubair, Vasilis K. Papanikolaou, and Robert Schober}

\IEEEauthorblockA{Institute for Digital Communications, Friedrich-Alexander-University Erlangen-Nuremberg, 91058 Erlangen, Germany}

\IEEEauthorblockA{E-mail: \{muhammad.zubair, vasilis.papanikolaou, robert.schober\}@fau.de}
\vspace{-10mm}

}

\maketitle
\vspace{-1cm}
\begin{abstract}
In this article, we develop an analytical radiation-pattern model for pinching-antenna systems (PASS) based on a two-dimensional dielectric slab waveguide. The model is derived in two steps. First, we employ coupled-mode theory (CMT) to derive a closed-form expression for the field coupled into the pinching antennas (PAs). Second, we use this analytical field profile as a scattering source model and derive the far-field radiation pattern via a two-dimensional radiation integral. We validate the proposed model against full-wave finite-element simulations performed in COMSOL Multiphysics, showing that it accurately reproduces the directional radiation characteristics of PASS. In contrast, most existing works model PAs as omni-directional point radiators, which simplifies system-level analysis but does not accurately capture the underlying electromagnetic radiation mechanism. Because the proposed model is given in closed form, it can be easily integrated into existing system-level PASS models to replace the assumed omni-directional pattern with a physically motivated directional radiation pattern. Finally, numerical simulations quantify the performance degradation that arises when the directional behavior of PAs is neglected in a representative wireless communications scenario.
\end{abstract}


\vspace{-0.1cm}
\section{Introduction}
\label{sec:intro}
Pinching-antenna systems (PASS), first demonstrated by NTT DOCOMO~\cite{Fukuda2022}, realize flexible antennas by applying small dielectric pieces, i.e., pinching antennas (PAs), at arbitrary positions along a low-loss dielectric waveguide, enabling reconfigurable radiation sites close to the user and strong line-of-sight links over distances that fixed-antenna and other flexible-antenna architectures cannot readily cover~\cite{Ding2025flexible}.

Performance analysis and resource-allocation frameworks have been developed for single- and multi-user PASS, targeting, e.g., non-orthogonal multiple access (NOMA) transmission~\cite{Ding2025flexible,Xu2025rate,Tegos2025uplink}, wireless-powered networks~\cite{Papanikolaou2025WPPAN}, and integrated sensing and communications~\cite{Khalili2025ISAC}.

A common assumption in existing PASS models is that each PA behaves as an omni-directional point radiator, so that only the coupled power and propagation distance determine the channel gain. Although the authors of \cite{PAModeling2025} apply coupled-mode theory (CMT) to characterize the power transferred from the waveguide to the PA, the radiated field is still treated as omni-directional. Consequently, the spatial radiation pattern, including the main-lobe direction, beamwidth, and sidelobe structure, remains unmodeled. While this simplification is convenient for system-level analysis, it is physically inconsistent with the underlying electromagnetic principles.
Recent work in~\cite{DiPASS2025} acknowledged this shortcoming and proposed a directional PASS framework (DiPASS) based on a Gaussian beam approximation. However, the model in~\cite{DiPASS2025} is parameterized empirically, does not capture the actual spatial variations of the coupled field along the PA length, and has not been validated against full-wave finite-element method (FEM) simulations.
 
Physically, a PA perturbs the guided mode of the dielectric waveguide, causes the coupling of a fraction of the guided energy into the PA, and subsequently radiates this energy into free space. This mechanism suggests that the PA does not behave as an omni-directional point source; rather, its radiation is expected to depend on the phase relation between the guided field and the radiated free-space waves and on the finite PA length. In particular, phase matching influences the dominant radiation direction, while the PA length affects the aperture size and thus impacts both the beam direction and the beamwidth. This highlights the need for a radiation-pattern model that can predict these directional characteristics and quantify the errors introduced by omni-directional approximations.

In this paper, we address this gap by developing a closed-form radiation-pattern model for a PASS based on a two-dimensional (2D) dielectric slab waveguide. The 2D setting reduces the radiation problem to a single angular variable and enables a fully analytical treatment of both the coupling and the radiation integrals. This allows us to isolate the essential physical mechanisms governing PA radiation, namely phase matching, power coupling, and finite-aperture radiation, without the additional modal complexity of full 3D waveguide structures, which generally admit only numerical or approximate analytical solutions~\cite{Book:okamoto2021fundamentalsOpticalWG,Book:kasap2013OptoElectronics}. The resulting model provides a tractable foundation for future extensions to 3D geometries.
 
The contributions of this paper are as follows. We derive, from first principles and without any fitting parameters, a closed-form radiation-pattern expression for a 2D PASS that captures the dependence of beam direction and beamwidth on waveguide geometry and operating frequency. The derivation combines CMT modeling of the coupled PA field with a 2D radiation integral over the PA. We validate the model against full-wave FEM simulations in COMSOL Multiphysics across a range of PA dimensions. The full-wave FEM simulations confirm the predicted directional behavior and show that neglecting the PA radiation pattern leads to systematic errors in directivity and beamwidth. Finally, we demonstrate that the resulting directional gain function can be directly substituted for the omni-directional pattern in existing system-level PASS frameworks~\cite{Ding2025flexible,Papanikolaou2025WPPAN,Khalili2025ISAC}, and we quantify the loss in spectral efficiency incurred by neglecting the PA's directivity.

\section{System Model}\label{s:2}
We consider the 2D dielectric slab waveguide shown in Fig.~\ref{fig:SysModel}. The main waveguide has length $L$, width $W_m$, core refractive index $n_1$, and cladding refractive index $n_0$. It is aligned along the $x$-axis and centered on $y = 0$, with input and output ports at $M_0 = (0, 0)$ and $M_L = (L, 0)$, respectively. Two identical dielectric PAs of length $L_s$ and width $W_s$ are attached symmetrically to the upper and lower surfaces of the main waveguide with their centers at $M_{u} = \left(x_{p}, y_u\right)$ and $M_{d} = \left(x_{p}, y_{d}\right)$, where $y_{u} =\frac{W_m + W_s}{2}$, $y_{d} = -\frac{W_m + W_s}{2}$  and $x_p \in \bigl[\tfrac{L_{s}}{2}, L - \tfrac{L_{s}}{2}\bigr]$. The point $M_p=(x_{p},y=0)$ represents the center of the waveguide where  the PAs are attached. Symmetric placement is adopted because it yields a radiation pattern symmetric with respect to the waveguide axis\footnote{With only a single, non-symmetric PA, we observe numerically that the main lobe is pushed toward end-fire radiation along the waveguide axis, which can only serve users near that axis. Moreover, to extract appreciable power a comparatively long PA is needed, which behaves as a waveguide coupler rather than a radiator.}, which provides the closest physically realizable counterpart to the omni-directional assumption adopted in prior PASS analyses~\cite{Ding2025flexible} and thus enables a fair comparison in Section~\ref{sec:numerical}. A user equipment (UE) is located at $M_{\mathrm{ue}} = (x_{\mathrm{ue}}, y_{\mathrm{ue}})$. The fundamental transverse electric mode $\mathrm{TE}_0$ is launched at $M_0$, and the output port at $M_L$ is assumed to be matched so that reflections can be neglected\footnote{This is implemented in the FEM simulations by a perfectly matched layer (PML).}.

\begin{figure}[h!]
	\centering
	\includegraphics[width=\linewidth]{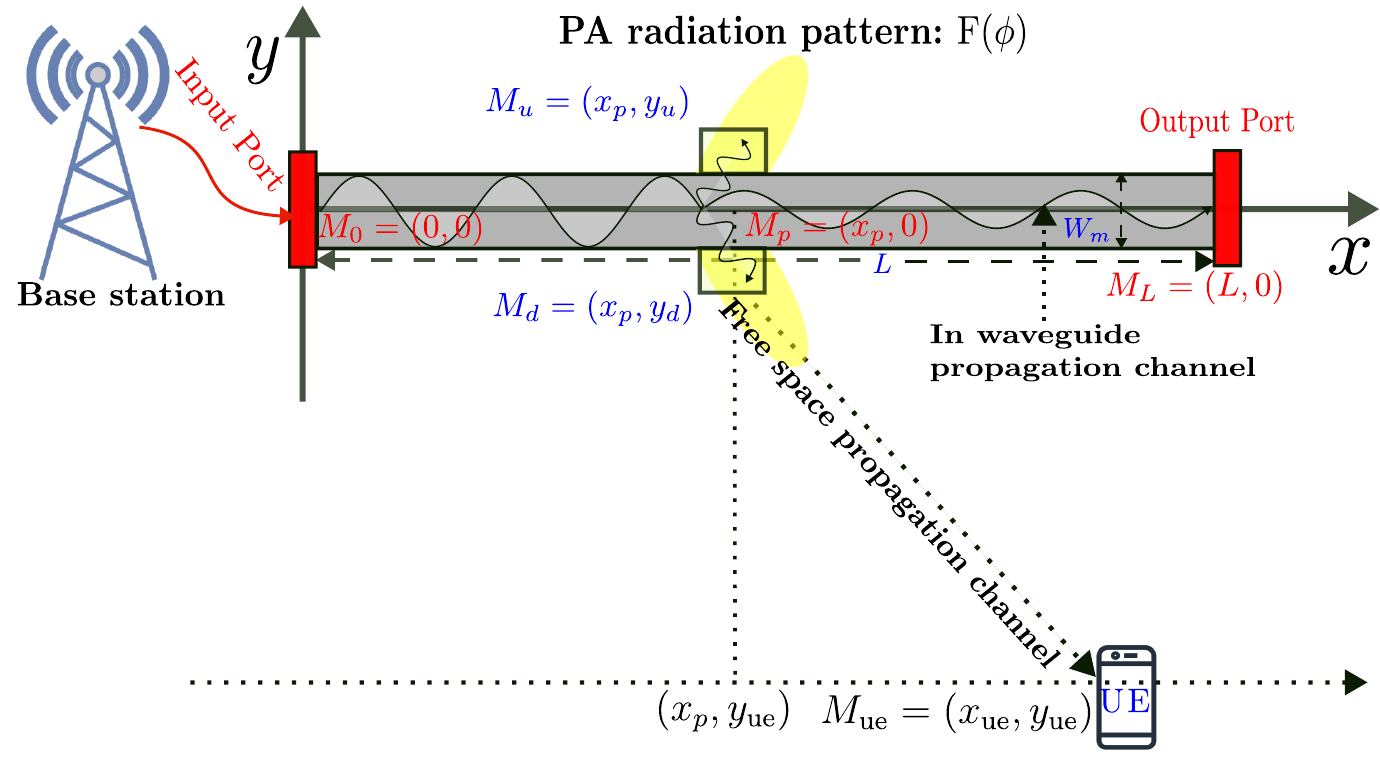}
	\vspace{-5mm}
	\caption{System model for the 2D  PASS.}
	\label{fig:SysModel}
	\vspace{-3mm}
\end{figure}

The guided wave is confined primarily to the core and decays exponentially in the cladding. Its transverse electric field profile is given by~\cite{Book:okamoto2021fundamentalsOpticalWG},
\begin{equation}
	E_{m}\left(y\right) =
	\begin{cases}
		A_{0} \cos(\beta_{my} y ), \quad -\frac{W_m}{2} \le y \le \frac{W_m}{2},\\
		A_{0} \cos(\beta_{my} \frac{W_m}{2} )  e^{-\sigma_{m} (|y|-\frac{W_m}{2})}, \quad |y| > \frac{W_m}{2},
	\end{cases}
	\label{eq:2DFields_MainGuide}
\end{equation}
Here, $A_0$ denotes a constant normalized amplitude, and $\beta_{my}=\sqrt{\beta_0^{2}n_{1}^{2}-\beta_{mx}^{2}}$ and $\sigma_m=\sqrt{\beta_{mx}^{2}-\beta_0^{2}n_{0}^{2}}$ represent the transverse propagation constant and the cladding decay constant of the main waveguide, respectively. In this context, $\beta_{mx}$ is the longitudinal propagation constant, and $\beta_0 = 2\pi/\lambda=2\pi f\sqrt{\mu_0\epsilon_0}$ is the free-space wavenumber, where $\lambda, f, \epsilon_{0}$, and $\mu_{0}$ are the wavelength, operating frequency, the vacuum permittivity, and the vacuum permeability, respectively.

The electric field profile in the PAs is obtained analogously. Let $s \in \{u, d\}$ index the upper and lower PAs, with centers $y_u$ and $y_d$. The transverse electric field profile of PA $s$ is
\begin{equation}
	E_{s}\left(y\right) =
	\begin{cases}
		B_{0s} \cos(\beta_{s y} (y - y_s ) ), \quad |y - y_s | \le \frac{W_s}{2},\\
		B_{0s } \cos(\beta_{s y} \frac{W_s}{2} )  e^{-\sigma_{s } (|y - y_s |-\frac{W_s}{2})}, \, |y-y_s | > \frac{W_s}{2},
	\end{cases}
	\label{eq:2DFields_Pinch}
\end{equation}
where $B_{0s}$ is a normalized constant amplitude, $\beta_{sx}$ and $\beta_{sy}$ are respectively the longitudinal and transverse propagation constants of PA $s$ with $\beta_{sy} = \sqrt{\beta_0^2 n_1^2 - \beta_{sx}^2}$, and $\sigma_s = \sqrt{\beta_{sx}^2 - \beta_0^2 n_0^2}$ is the evanescent decay constant.

For each of the three slabs, the full guided field propagating in the $x$-direction can be written as
\begin{equation}
	\tilde{E}_q(x,y)=E_q(y)e^{-j\beta_{qx}x}, \qquad q\in\{m,u,d\}.
	\label{eq:totalField}
\end{equation}
where $q \in \{m, u, d\}$ indexes the main waveguide and the two PAs.
The propagation and cladding decay constants in (\ref{eq:2DFields_MainGuide}) and (\ref{eq:2DFields_Pinch}) are determined by the dispersion relations of each isolated slab. For slab $q \in \{m, u, d\}$ with width $W_q$, they satisfy $\beta_{qy}^2 + \beta_{qx}^2 = \beta_0^2 n_1^2$ and $\sigma_q^2 + \beta_0^2 n_0^2 = \beta_{qx}^2$. Enforcing field continuity at the core-cladding interfaces yields the even-mode eigenvalue equation $\beta_{qy} \tan\!\left(\beta_{qy} W_q / 2\right) = \sigma_q$. With the normalized variables $u_q = \beta_{qy} W_q / 2$, $w_q = \sigma_q W_q / 2$, and the normalized frequency $V_q = (\beta_0 W_q / 2)\sqrt{n_1^2 - n_0^2}$, this becomes
\begin{align}
	u_q\tan(u_q) &= w_q, \label{eq:guide_uv_eq1}\\
	u_q^2+w_q^2 &= V_q^2. \label{eq:guide_uv_eq2}
\end{align}
For a given geometry and operating frequency, \eqref{eq:guide_uv_eq1} and \eqref{eq:guide_uv_eq2} determine the modal propagation constants of the isolated slabs. The normalized frequency $V_q$ governs the number of guided modes supported by slab $q$; to ensure that each slab operates in the single-mode regime, which is the case considered in this work, the slab parameters are chosen such that $V_q<\pi/2$~\cite[Sec.~2.1, Eq.~2.17]{Book:kasap2013OptoElectronics}.
\section{Physics-Based Modeling of PA and its Radiation Pattern}
In this section, we develop the proposed analytical radiation-pattern model. We first formulate the coupled-mode equations for the three-slab system, derive closed-form expressions for the coupling coefficient and for the field coupled into the PAs, and then use this field for the source in a 2D radiation integral to obtain the far-field radiation pattern and the corresponding channel model.

\subsection{Coupled-Mode Equations}
The isolated guided modes of the main waveguide and both PAs are given by $\tilde{E}_q(x,y)$ in~\eqref{eq:totalField} for $q \in \{m, u, d\}$, with the transverse profiles $E_m(y)$ and $E_s(y)$ defined in~\eqref{eq:2DFields_MainGuide} and~\eqref{eq:2DFields_Pinch}. 

When the two PAs are attached to the main waveguide, the three isolated modes interact. Following the CMT framework, the total field is expressed as a linear combination of the isolated-slab modes~\cite{Book:okamoto2021fundamentalsOpticalWG}
\begin{equation}
	E(x,y) = A(x)\tilde{E}_m(x,y) + B_u(x)\tilde{E}_u(x,y) + B_d(x)\tilde{E}_d(x,y),
	\label{eq:supermode}
\end{equation}
where $A(x)$, $B_u(x)$, and $B_d(x)$ are slowly varying mode amplitudes. Substituting~\eqref{eq:supermode} into the standard CMT formulation~\cite{PAModeling2025,YarivCMT} yields the following coupled equations for the mode amplitudes,
\begin{subequations}
	\begin{align}
		\frac{\mathrm{d}A(x)}{\mathrm{d}x}   &= -j\kappa_{mu} B_u(x) e^{-j\Delta_u x}
		-j\kappa_{md} B_d(x) e^{-j\Delta_d x},  \label{eq:cmt_general1}\\
		\frac{\mathrm{d}B_u(x)}{\mathrm{d}x} &= -j\kappa_{um} A(x) e^{j\Delta_u x}, \label{eq:cmt_general2}\\
		\frac{\mathrm{d}B_d(x)}{\mathrm{d}x} &= -j\kappa_{dm} A(x) e^{j\Delta_d x},\label{eq:cmt_general3}
	\end{align}
	\label{eq:cmt_general}
\end{subequations}
\par\noindent where $\Delta_u = \beta_{ux} - \beta_{mx}$ and $\Delta_d = \beta_{dx} - \beta_{mx}$ are the phase mismatches between the main waveguide and the upper and lower PAs, respectively. The coefficients $\kappa_{mu}$ and $ \kappa_{md}$ characterize the coupling from the main waveguide into the upper and lower PAs, and $\kappa_{um}$ and $\kappa_{dm}$ characterize the reverse coupling from the upper and lower PAs, respectively.
\subsection{Symmetric Phase-Matched Case}
For the problem at hand, the two PAs are identical in geometry and are placed symmetrically relative to a main waveguide. Thus $\sigma_s = \sigma_u = \sigma_{d}$ and the coupling coefficients for the upper and lower PAs are identical by symmetry, i.e., $\kappa_{mu} = \kappa_{md}\triangleq \kappa_m$. In the following, Theorem~\ref{theorem:kappa} establishes this relationship and provides a closed-form expression for the common coupling coefficient $\kappa_m$.
\begin{theorem}
\label{theorem:kappa}
Under symmetric conditions, the coupling coefficient $\kappa_m$ between the main waveguide and the PAs is given by \eqref{eq:kappac}, shown at the top of the page.
\begin{figure*}[ht!]
\begin{equation}
\label{eq:kappac}
\kappa_m = \frac{\beta_{0}^{2}\left(n_{1}^{2} - n_{0}^{2}\right)}{2\beta_{mx}} e^{-\sigma_s W_{m}/2}\cos\left( \frac{\beta_{sy} W_{s}}{2}\right) 
\frac{2\sigma_{s} \cos\!\left(\frac{\beta_{my} W_{m}}{2}\right)
\sinh\!\left(\frac{\sigma_s W_{m}}{2}\right)
+
2\beta_{my} \sin\!\left(\frac{\beta_{my} W_{m}}{2}\right)
\cosh\!\left(\frac{\sigma_{s} W_{m}}{2}\right)}{\left(\sigma_{s}^{2} + \beta_{my}^{2}\right)  \left(\frac{W_{m}}{2} + \frac{\sin\left(\beta_{my} W_{m}\right)}{2\beta_{my}} + \frac{\cos^{2}\left(\frac{\beta_{my} W_{m}}{2}\right)}{\sigma_m}\right) }. 
\end{equation}
\hrule
\vspace{-5mm}
\end{figure*}
\end{theorem}
\begin{IEEEproof}
Please refer to Appendix~\ref{app:coupling}.
\end{IEEEproof}
\begin{remark}
Theorem~\ref{theorem:kappa} holds for any symmetric two-PA configuration and does not require the PAs and main waveguide to share the same propagation constants. 
\end{remark}

In a similar manner based on geometrical symmetry and Theorem~\ref{theorem:kappa}, the reverse coupling coefficients are related as $\kappa_{um} = \kappa_{dm} \triangleq \kappa_s$. Since the two PAs are identical in geometry, possess the same material properties, and are placed symmetrically with respect to the main waveguide, their propagation constants are equal. Hence, $\beta_{ux} = \beta_{dx}$, and the PA amplitudes are identical, i.e., $B_u(x) = B_d(x) \triangleq B(x)$. Under these conditions and using the upper-lower PA symmetric coupling coefficients, \eqref{eq:cmt_general} simplifies to
\begin{subequations}
\begin{align}
\frac{\mathrm{d}A(x)}{\mathrm{d}x} &= -2j\kappa_m B(x)e^{-j\Delta x}, \label{eq:cmt_symmetric1}\\
\frac{\mathrm{d}B(x)}{\mathrm{d}x} &= -j\kappa_s A(x)e^{j\Delta x},\label{eq:cmt_symmetric2}
\end{align}
\label{eq:cmt_symmetric}
\end{subequations} 
\par\noindent where $\Delta \triangleq \Delta_u = \Delta_d$. In the remainder of this work, we further impose phase matching between the main waveguide and the PAs, i.e., $\beta_{mx} = \beta_{ux} = \beta_{dx} \triangleq \beta_{sx}$ (so that $\Delta = 0$), which is achieved by adopting identical widths $W_m = W_s$. Reciprocity then yields $\kappa_m = \kappa_s^{*}\triangleq\kappa$, where $(\cdot)^{*}$ denotes complex conjugation~\cite{Book:okamoto2021fundamentalsOpticalWG}. Under these conditions,~\eqref{eq:cmt_symmetric} becomes
\begin{subequations}
	\begin{align}
		\frac{\mathrm{d}A(x)}{\mathrm{d}x} &= -2j\kappa B(x), \label{eq:coupledModeMain1}\\
		\frac{\mathrm{d}B(x)}{\mathrm{d}x} &= -j\kappa A(x). \label{eq:coupledModeMain2}
	\end{align}
	\label{eq:coupledModeMain}
\end{subequations}
\vspace{-0.5em}
\begin{theorem} 
	\label{theorem:AxBx}
	Let $A_0$ denote the main-waveguide amplitude at the start of the coupling region, and define the local coordinate $\tilde{x} = x - (x_{p} - L_s/2) \in [0, L_s]$. With boundary conditions $A(\tilde{x} = 0) = A_0$ and $B(\tilde{x} = 0) = 0$, the solution of~\eqref{eq:coupledModeMain} is given by
	\begin{align}
		A(\tilde{x}) = A_{0} \cos\!\left(\sqrt{2}\kappa \tilde{x}\right), \,\,\,B(\tilde{x}) = -\frac{j A_0}{\sqrt{2}} \sin\!\left(\sqrt{2}\kappa \tilde{x}\right).
		\label{eq:modeAmplitudes}
	\end{align}
\end{theorem}
\begin{IEEEproof}
Please refer to Appendix~\ref{app:axbx}.
\end{IEEEproof}

By Theorem~\ref{theorem:AxBx}, $A(x)$ over the entire waveguide length in the global coordinate system is given by
\begin{equation}\small
  A(x)\! =\!
  \begin{cases}
    A_0,
    & 0 \leq x < x_{p} - \frac{L_s}{2},\!\!
    \\[4pt]
    A_0\cos\!\left(\sqrt{2}\kappa\!\left(x-x_{p} + \frac{L_s}{2}\right)\right),\!
    &\!\! x_{p}\! -\! \frac{L_s}{2}\! \leq \!x\! \leq x_{p}\! +\! \frac{L_s}{2},
    \\
    A_0\cos\!\left(\sqrt{2}\kappa L_s\right),
    &\!\! x_{p} + \frac{L_s}{2} < x \leq L.
  \end{cases}
  \label{eq:Ax}
\end{equation}
Amplitude $B(x)$ is nonzero only within the coupling region $x_{p} - \frac{L_s}{2}  \leq x \leq x_{p} + \frac{L_s}{2}$, where it is given by
\begin{equation}\small
  B(x) = 
    -\frac{j}{\sqrt{2}} A_0
    \sin\!\left(\sqrt{2} \kappa\!\left(x-x_{p} + \frac{L_s}{2}\right)\right).
  \label{eq:Bx}
\end{equation}
Since $B(x)$ modulates the transverse profile of both PAs, the evanescent tail of the PA mode, which extends beyond the waveguide cladding into free space, carries a spatially varying amplitude that acts as a finite radiating aperture. Exploiting~\eqref{eq:totalField} and \eqref{eq:Bx}, the resulting field along each PA $s$ is obtained as
\begin{align}
    \tilde{E}_{s}(x,y)\!
             =\! - \frac{jA_0}{\sqrt{2}}\sin\!\left(\sqrt{2} \kappa\left(x-x_{p} \!+\! \frac{L_s}{2}\right)\!\right)\!\!E_{s}(y)e^{-j\beta_{sx} x},
           \label{eq:FieldInPinch}
\end{align}
whereas the field of the wave in the main waveguide after the coupling region is given by
\begin{align}
 &   \tilde{E}_{m}(x,y) = A(x) E_m(y)e^{-j\beta_{mx}x} \notag\\
           &  = A_0\cos\!\left(\sqrt{2} \kappa L_{s}\right)E_m(y)e^{-j\beta_{mx} x}.
           \label{eq:FieldInMG}
\end{align}
\subsection{Power Coupling}
The coupling coefficient $\kappa$ determines the fraction of power transferred from the main waveguide to the PAs. Using $\kappa$, the normalized power in the main waveguide and in each PA can be written as
\begin{equation}
    P_{m} = \frac{|A(\tilde{x})|^2}{A^{2}_{0}} = \cos^{2}\!\left(\sqrt{2}\kappa \tilde{x}\right),
    \label{eq:PowerM}
\end{equation}
\begin{equation}
    P_{s} = \frac{|B(\tilde{x})|^2}{A^{2}_{0}} = \frac{\sin^2\!\left(\sqrt{2}\kappa \tilde{x}\right)}{2}.
    \label{eq:PowerP}
\end{equation}

\begin{corollary}
\label{corollary:CouplingLength}
The coupling length, i.e., the shortest PA length at which complete power transfer from the main waveguide to the PAs occurs, is $L_c = \frac{\pi}{2\sqrt{2}\kappa}$.
\end{corollary}
\begin{IEEEproof}
From \eqref{eq:PowerP}, $P_{s}$ reaches its maximum when
\begin{equation}
    \tilde{x} = \frac{\pi(2m+1)}{2\sqrt{2}\kappa}, \qquad m = 0,1,2,\ldots
\end{equation}
Setting $m=0$ yields $L_c = \frac{\pi}{2\sqrt{2}\kappa}$. %
\end{IEEEproof}

\subsection{Radiation Pattern and Channel Model}
In this subsection, we derive the far-field radiation pattern of the PAs and formulate the corresponding channel model.

\begin{theorem}
	\label{theorem:RadPattern}
	Given the electric field distribution $\tilde{E}_{s}(x,y)$ in \eqref{eq:FieldInPinch}, the PA far-field radiation pattern can be expressed as
	\begin{equation}
		F(\phi) = P(\phi)F_x(\phi)F_y(\phi),
		\label{eq:RadPattern}
	\end{equation}
	where $P(\phi) = -\sin(\phi)$ is the projection factor, and $F_x(\phi)$ and $F_y(\phi)$ are the longitudinal and transverse radiation pattern components, respectively, given by
	\begin{equation}\small
		F_x(\phi)\!=\! -\frac{L_s A_{0}}{2\sqrt{2}}e^{j\left(k_x - \beta_{sx}\right)x_{p}}
		\left(\!
		e^{j\delta_x}
		\mathrm{sinc}\!\left(\Omega^{+}_{x}\right)\! -\!
		e^{-j\delta_x}\mathrm{sinc}\!\left(\Omega^{-}_{x}\right)
		\!\right),
		\label{eq:Fx}
	\end{equation}
	\begin{equation}\small
		F_y(\phi)\! =\! \tfrac{W_s}{2}(
		e^{j\delta_{y}^{+}-j \beta_{sy} y_s}\mathrm{sinc}\!(\Omega^{+}_{y})
		\!+\!
		e^{j\delta_{y}^{-} + j\beta_{sy} y_s}\mathrm{sinc}\!(\Omega^{-}_{y})),
		\label{eq:Fy}
	\end{equation}
	where $\sinc(x) = \frac{\sin(x)}{x}$,
$\Omega^{+}_{x} = \frac{L_s}{2}\left(k_x - \beta_{sx} + \sqrt{2}\kappa\right)$, $\Omega^{-}_{x} = \frac{L_s}{2}\left(k_x - \beta_{sx} - \sqrt{2}\kappa\right)$, $\delta_{x} = \frac{\kappa L_{s}}{\sqrt{2}}$, $k_x = \beta_0\cos\phi$, and  $\Omega^{+}_{y} = \frac{W_s}{2}\left(k_y + \beta_{sy} \right)$, $\Omega^{-}_{y} = \frac{W_s}{2}\left(k_y - \beta_{sy}\right)$, $\delta_{y}^{+} = \left(k_y + \beta_{sy}\right)(W_m + W_s)$, $\delta_{y}^{-} = \left(k_y - \beta_{sy}\right)(W_m + W_s)$, $k_y = \beta_0\sin\phi$.
\end{theorem}

\begin{IEEEproof}
    Please refer to Appendix~\ref{app:rp}.
\end{IEEEproof}

The normalized PA power pattern is defined as $G\left(\phi\right) = \tfrac{|F(\phi)|^2}{\max{|F(\phi)|^2}}$~\cite[Sec. 2.4.5, Eq. 2-119]{Book:stutzman2012antennaTheoryAndDesign}, so its maximum value is one. Exploiting~\cite[Sec. 2.5, Eq. 2-145]{Book:stutzman2012antennaTheoryAndDesign}, the 2D directivity of the PA can be defined as
\begin{equation}
	D(\phi) = \frac{G(\phi)}{\frac{1}{2\pi}\displaystyle\int_0^{2\pi}G(\phi)\mathrm{d}\phi}.
	\label{eq:directivity}
\end{equation}
Using \eqref{eq:directivity} and the power ratio $P_s$ from (\ref{eq:PowerP}), the channel between the waveguide input port and the UE can be expressed as
\begin{equation}
	h = \sqrt{\frac{\eta\,2P_{s} D(\phi)}{r}} e^{-j\left(\frac{2\pi}{\lambda_g}r_p + \frac{2\pi}{\lambda}r\right)},
	\label{eq:proposedChan}
\end{equation}
where $\eta = {\frac{c}{4\pi^2 f}}$ is the 2D free-space path-loss constant, $1/\sqrt{r}$ is used to account for 2D free-space propagation\cite{Book:balanis2011modernAntennaHandbook}, $c$ is the speed of light, $r_p = \|M_p - M_0\|$ is the in-waveguide propagation distance with guided wavelength $\lambda_g=2\pi/\beta_{mx}$ up to the position of the PAs $M_p$, and $r = \|M_{\mathrm{ue}} - M_p\|$ is the free-space distance from the PA to the UE. In \eqref{eq:proposedChan}, $2P_s$ denotes the total fraction of power coupled from the main waveguide into both PAs. Furthermore, we assume that the PAs are lossless, so that all coupled power is radiated into free space, i.e., the radiation efficiency $e_r = P_\mathrm{rad}/\left(2P_s\right) \approx 1$~\cite[Sec. 2.5]{Book:stutzman2012antennaTheoryAndDesign}, where $P_\mathrm{rad}$ is the total radiated power.

\section{Numerical Results}
\label{sec:numerical}
To validate the proposed model, we carry out FEM simulations using the Electromagnetic Waves, Frequency Domain (EWFD) interface of COMSOL Multiphysics, which solves Maxwell's equations in the frequency domain. The PASS geometry is discretized with a mesh resolution of 12 elements per wavelength. The simulated PASS follows the architecture presented in  Section~\ref{s:2}, operating in the single-mode regime. A PML boundary encloses the two PAs and a central section of the main waveguide of length $\mathcal{L}_{mS}$
\footnote{$\mathcal{L}_{mS}$ is the length of the main-waveguide section enclosed by the PML, chosen large enough to capture the near-field and far-field behavior around the PAs. The portions of the waveguide outside this section do not contribute to the simulated scattered field.}, absorbing outgoing scattered fields without reflection. The simulation parameters are summarized in Table~\ref{table:SimParams}.

For performance assessment and validation of the proposed model, the normalized power patterns $G\left(\phi\right)$, obtained with the proposed model and from COMSOL\footnote{
The radial component of the Poynting vector functionality of COMSOL is employed for comparison~\cite{Book:stutzman2012antennaTheoryAndDesign}. Its far-field radial component in terms of proposed power pattern can be expressed as $S_{r}\left(\phi\right)  = \tfrac{\eta_{0} G(\phi)}{2}$, where $\eta_{0}$ is the free-space impedance.}, respectively,
are presented in Fig.~\ref{fig:RadiationPattern}. The results indicate that the proposed closed-form model agrees well with the COMSOL results for different PA lengths. Fig.~\ref{fig:RadiationPattern} reveals that for shorter PAs the radiation pattern has a larger beamwidth and lower side lobes than for longer PAs. Moreover, as the PA length is increased from $L_s = 0.75\lambda$ to $L_s = 2.5\lambda$, the radiation pattern is steered toward the end-fire direction, with higher directivity and more pronounced side lobes. Additionally, Fig.~\ref{fig:RadiationPatternWidth} presents the radiation pattern for a fixed PA length of $2\lambda$, where the widths $W_q = 0.408\lambda$ and $W_q = 0.470\lambda$ correspond to $V_q=1.35$ and $V_q=1.55$, respectively. The plots clearly show that the proposed model agrees closely with the COMSOL results.

Fig.~\ref{Fig:coupledpower} shows the fraction of power coupled into the PAs as a function of $L_s$ for both the proposed symmetric-PA configuration and a single-PA reference case. In both cases, the analytical and COMSOL results are in good agreement and exhibit the expected oscillatory power transfer. The considered symmetric configuration achieves maximum coupling at $L_s = L_c \approx 10$~mm ($2P_s \approx 1$), with full power return at $L_s = 2L_c \approx 20$~mm, whereas the single-PA case requires a longer coupling length since only one PA extracts power. A slight discrepancy between the analytical and COMSOL results is visible at larger $L_s$, because the theoretical model assumes all coupled power returns to the main waveguide, whereas COMSOL accounts for the fact that part of the power radiates into free space. The electric field norm of the PASS, simulated in COMSOL, is presented as surface plots in Figs.~\ref{Fig:comsolfield1} and~\ref{Fig:comsolfield2}. These figures show that when $L_s = 2\lambda$, the PAs radiate strongly, whereas for $L_s = 4\lambda$, most of the energy remains confined within the waveguide. This confirms the results in Fig.~\ref{Fig:coupledpower}.

\begin{table} 
\linespread{.9}\selectfont
\centering
 \caption{Simulation parameters. $ T_\mathrm{PML}$ and $D_\mathrm{PML}$ denote the thickness and the diameter of the circular boundary of the PML, respectively.
 }
\label{table:SimParams}
\setlength{\tabcolsep}{3pt}
\begin{tabular}{c | c || c | c || c | c} 
 Parameter & Value & Parameter & Value & Parameter & Value \\ [0.5ex] 
 \hline\hline
 $f_0$ & $60$ GHz & $n_1$ & $\sqrt{2.1}$ & $n_{0}$ & $1$ \\
 \hline
 $\mathcal{L}_{mS}$  & $150\,\mathrm{mm}$ & $D_\mathrm{PML}$ & $150\,\mathrm{mm}$ & $T_\mathrm{PML}$ & $1\,\mathrm{\lambda}$ \\
 \hline
 $P_\mathrm{T}$ & $1 \mathrm{W}$ & Mesh elements per $\lambda$ & $12$ & $\lambda$ & $5\,\mathrm{mm}$ 
 \end{tabular}
\vspace{-4mm}
\end{table}

\begin{figure*}[htbp]
    \begin{subfigure}[h]{0.2\linewidth}  
        \centering	
	 \includegraphics[width=\linewidth]{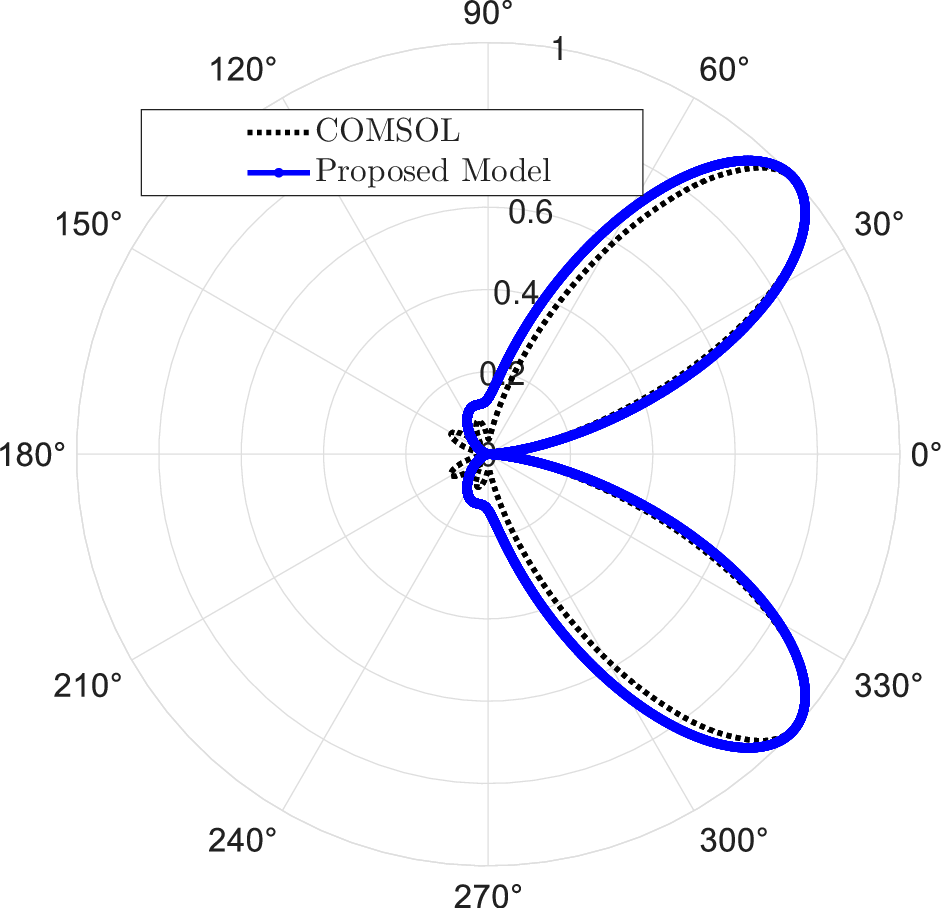}
	 \caption{$L_s= 0.75\lambda$}\label{Fig:radiationpatternpl1}
    \end{subfigure}
   \hfill 
    \begin{subfigure}[h]{0.2\linewidth}
         \centering
       \includegraphics[width=\linewidth]{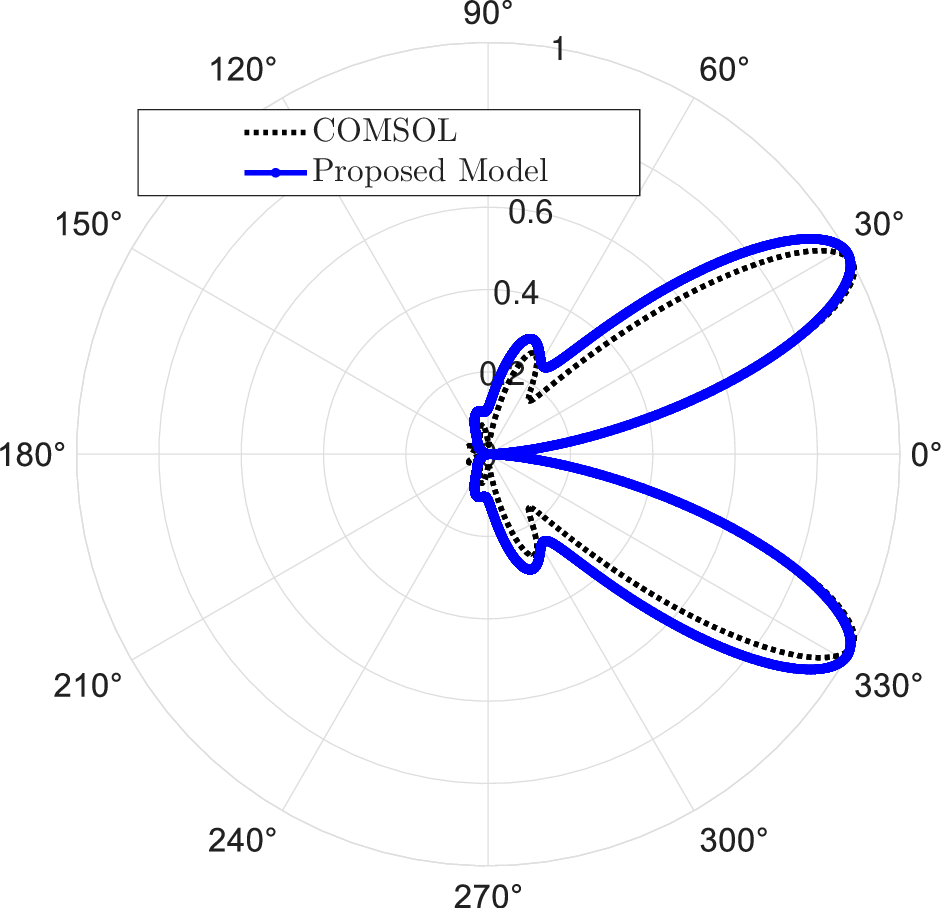}
  \caption{$L_s= 1.5\lambda$}\label{Fig:radiationpatternpl2}
    \end{subfigure}
    \hfill
    \begin{subfigure}[h]{0.2\linewidth}
      \centering
       \includegraphics[width=\linewidth]{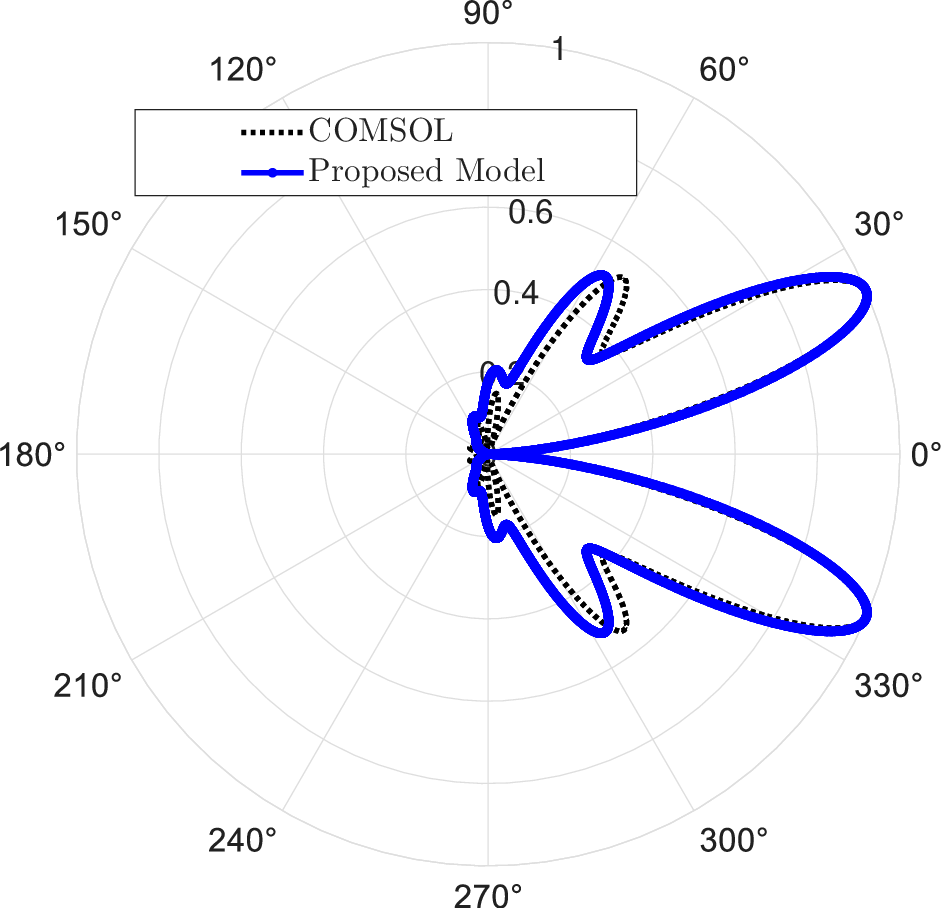}
\caption{$L_s= 2\lambda$}\label{Fig:radiationpatternpl3}
    \end{subfigure}
    \hfill 
    \begin{subfigure}[h]{0.2\linewidth}
         \centering
       \includegraphics[width=\linewidth]{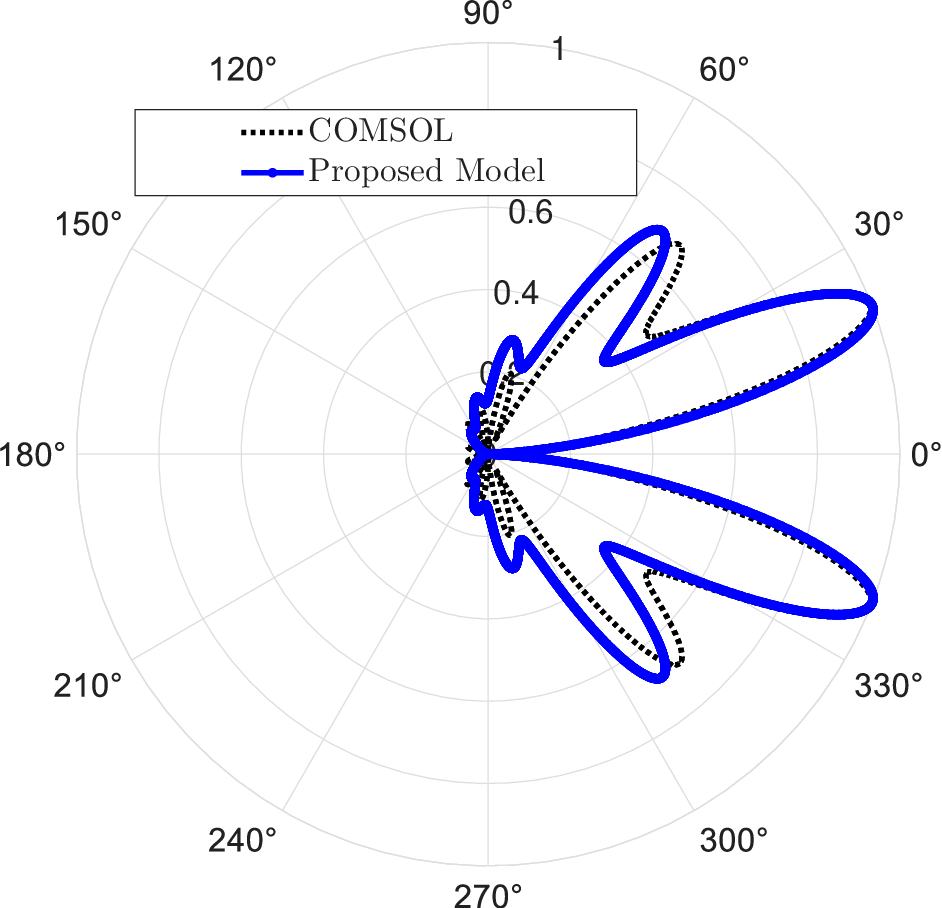}
       \caption{$L_s= 2.5\lambda$}\label{Fig:radiationpatternpl4}
    \end{subfigure}
    \caption{Normalized radiation patterns for $V_q = 1.5$.}
    \label{fig:RadiationPattern}
\end{figure*}

\begin{figure}[htbp]
    \begin{subfigure}[h]{0.4\linewidth}
         \centering
       \includegraphics[width=\linewidth]{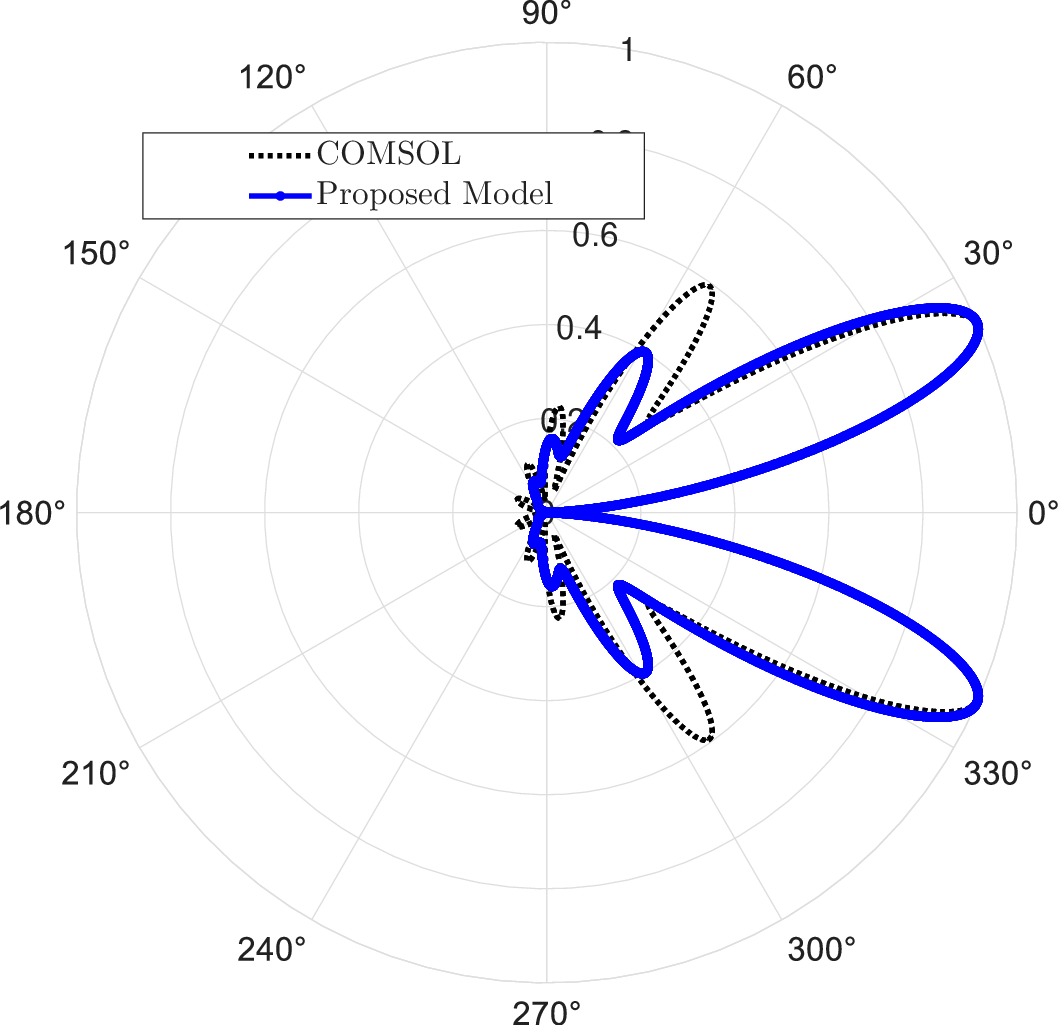}
  \caption{ $W_q= 0.408\lambda$}\label{Fig:width1}
    \end{subfigure}
    \hspace{0.02\columnwidth}
    \begin{subfigure}[h]{0.4\linewidth}
      \centering
       \includegraphics[width=\linewidth]{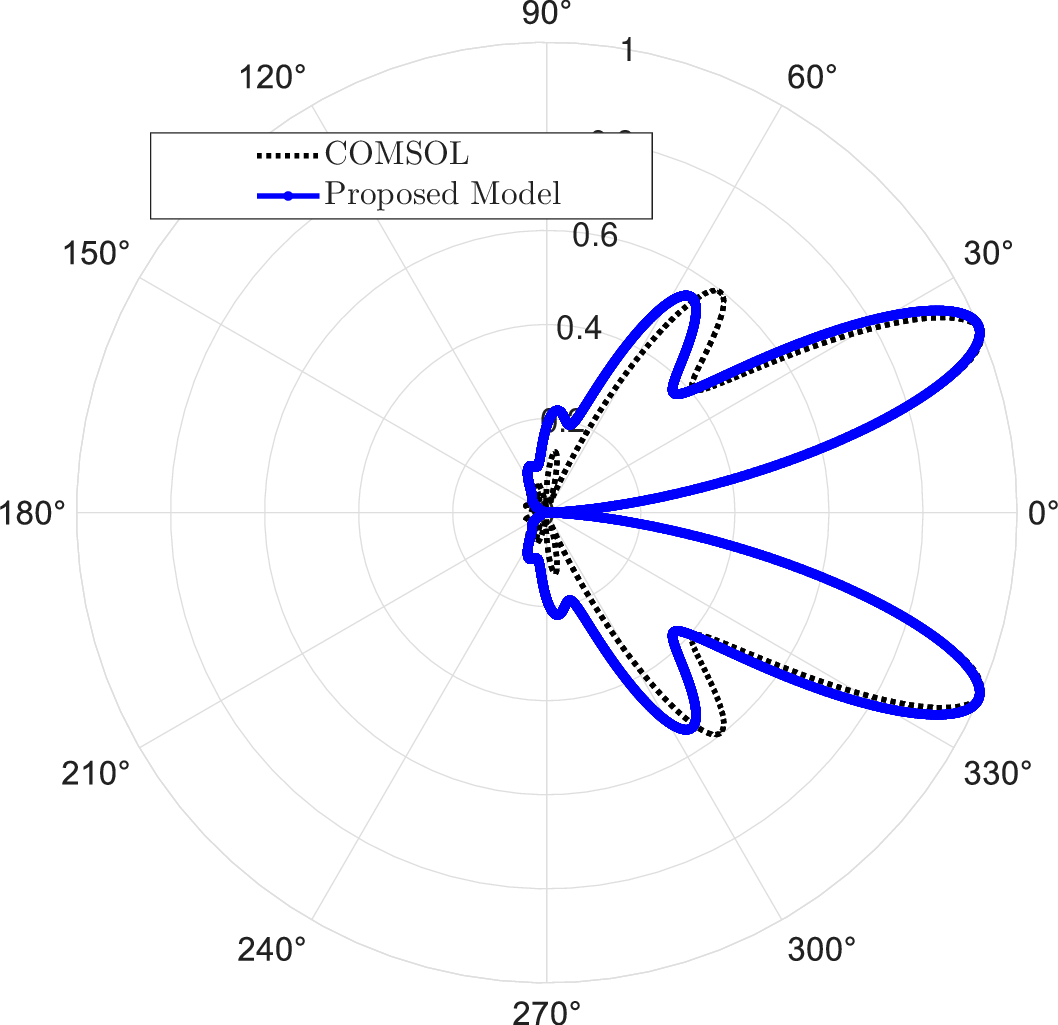}
\caption{ $W_q= 0.470\lambda$}\label{Fig:width2}
    \end{subfigure}
    \caption{Normalized radiation patterns for different widths at a fixed PA length of $2\lambda$.}
   \label{fig:RadiationPatternWidth}
   \vspace{-2mm}
\end{figure}
\begin{figure*}[htbp]
    \begin{subfigure}[h]{0.32\linewidth}  
        \centering	
	 \includegraphics[width=\linewidth]{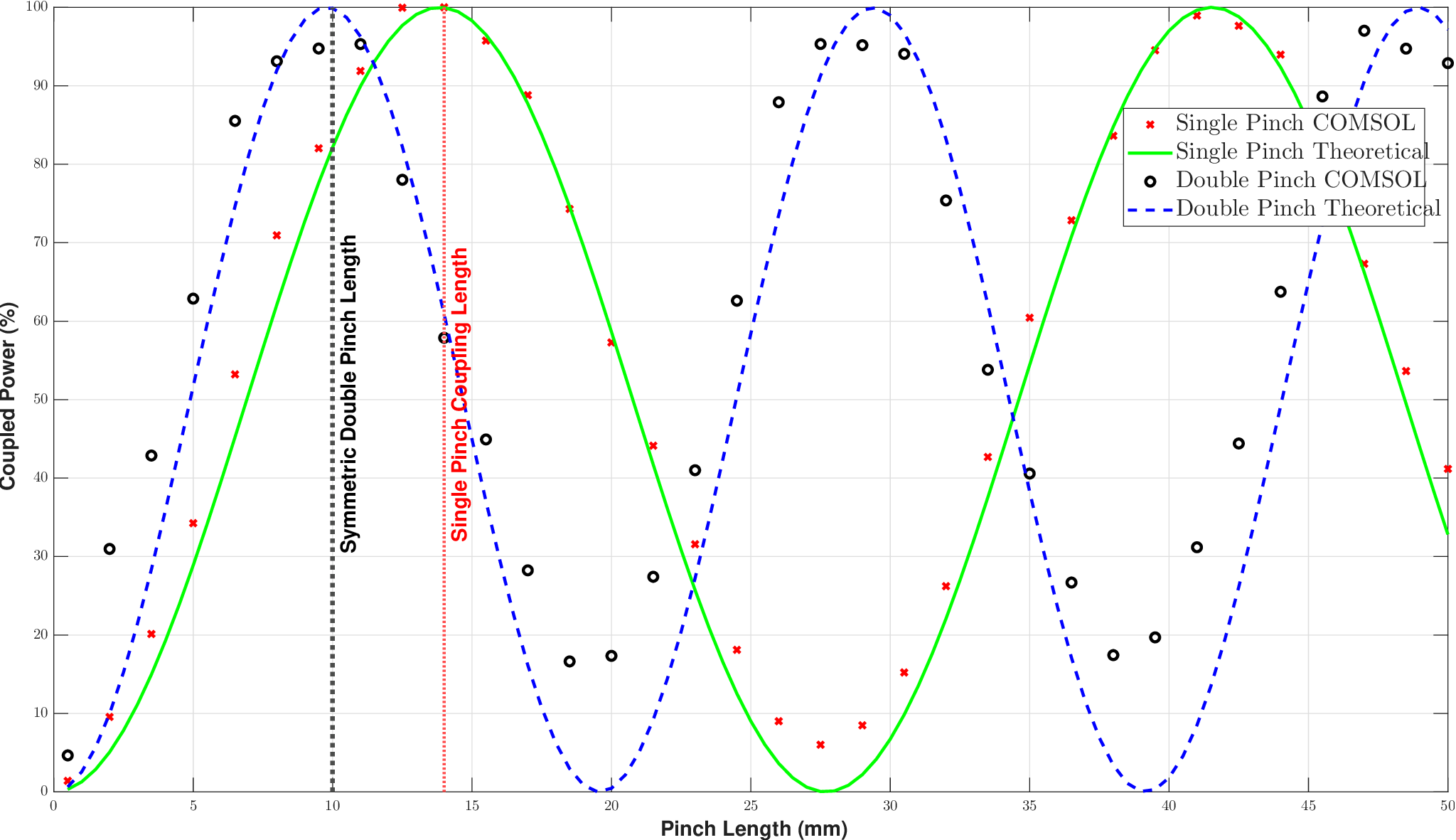}
	 \caption{Coupled power (\%) and coupling length of single PA and proposed design versus PA length for $V_q = 1.5$.}\label{Fig:coupledpower}
    \end{subfigure}
   \hfill 
    \begin{subfigure}[h]{0.29\linewidth}
         \centering
       \includegraphics[width=0.8\linewidth]{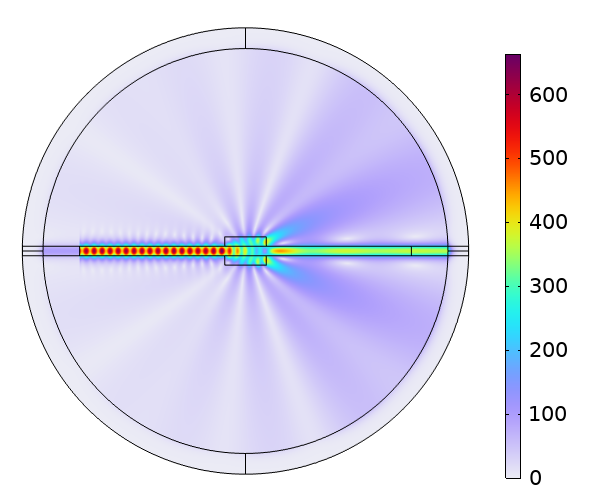}
  \caption{Norm of electric field for $L_s= 2\lambda$.}\label{Fig:comsolfield1}
    \end{subfigure}
    \hfill
    \begin{subfigure}[h]{0.29\linewidth}
      \centering
       \includegraphics[width=0.8\linewidth]{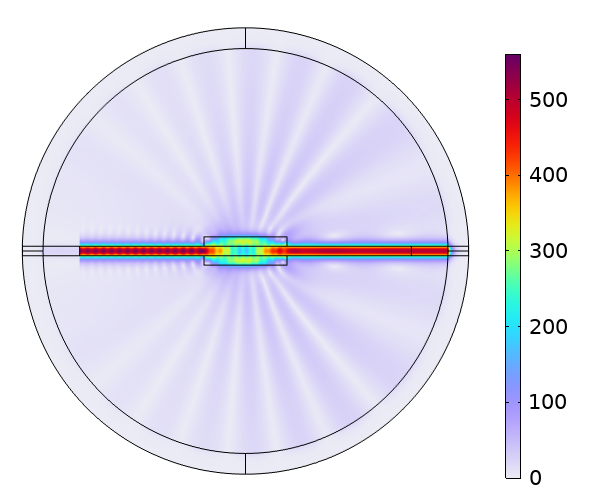}
\caption{ Norm of electric field for $L_s= 4\lambda$.}\label{Fig:comsolfield2}
    \end{subfigure}
    \caption{Power coupling and scattering characteristics of PAs with varying lengths.}
    \label{fig:Power}
    \vspace{-1em}
\end{figure*}
To evaluate the impact of the proposed model on communication performance, the spectral efficiency is used as a performance metric. We consider a PASS with a main waveguide of length $40$ m and width $2.2~\mathrm{mm}$ (corresponding to $V_q=1.5$), deployed at a height of $y_{\mathrm{ue}}=5~\mathrm{m}$ above the UE, see Fig. \ref{fig:SysModel}. For the Monte Carlo simulations, we draw $10^4$ independent realizations of the UE horizontal position, $x_{\mathrm{ue}}$, which is uniformly distributed between $0\,\mathrm{m}$ and $40\,\mathrm{m}$. Using the channel model proposed in (\ref{eq:proposedChan}), the spectral efficiency for symmetric PASS can be expressed as $R = \log\left(1 + |h|^2 P_T/\sigma^2\right)$, where $P_T$ is the input power to the main waveguide and $\sigma^2$ is the power of the zero-mean Gaussian noise. We define the transmit SNR as $\mathrm{SNR}_\mathrm{tr}= \frac{P_T}{\sigma^2}$. 
Fig.~\ref{fig:InfoRate} shows the average spectral efficiency versus the transmit SNR. For each scheme, the PA position along the $x$-axis is optimized by maximizing $R$ via grid search using the respective channel model. The spectral efficiency shown in Fig. \ref{fig:InfoRate} is then evaluated at the resulting optimal position using COMSOL. Fixed-antenna and PA positions optimized based on an omni-directional PA radiation pattern lose approximately $8$~dB and $7.5$~dB compared to the COMSOL-optimized positions, respectively, whereas the proposed model-optimized positions remain within $\sim \!0.4$~dB of the COMSOL upper bound, confirming the accuracy of the proposed model and the importance of accounting for the PA radiation pattern for PASS design.

\begin{figure}[h!]
    \centering
    \includegraphics[width=\linewidth]{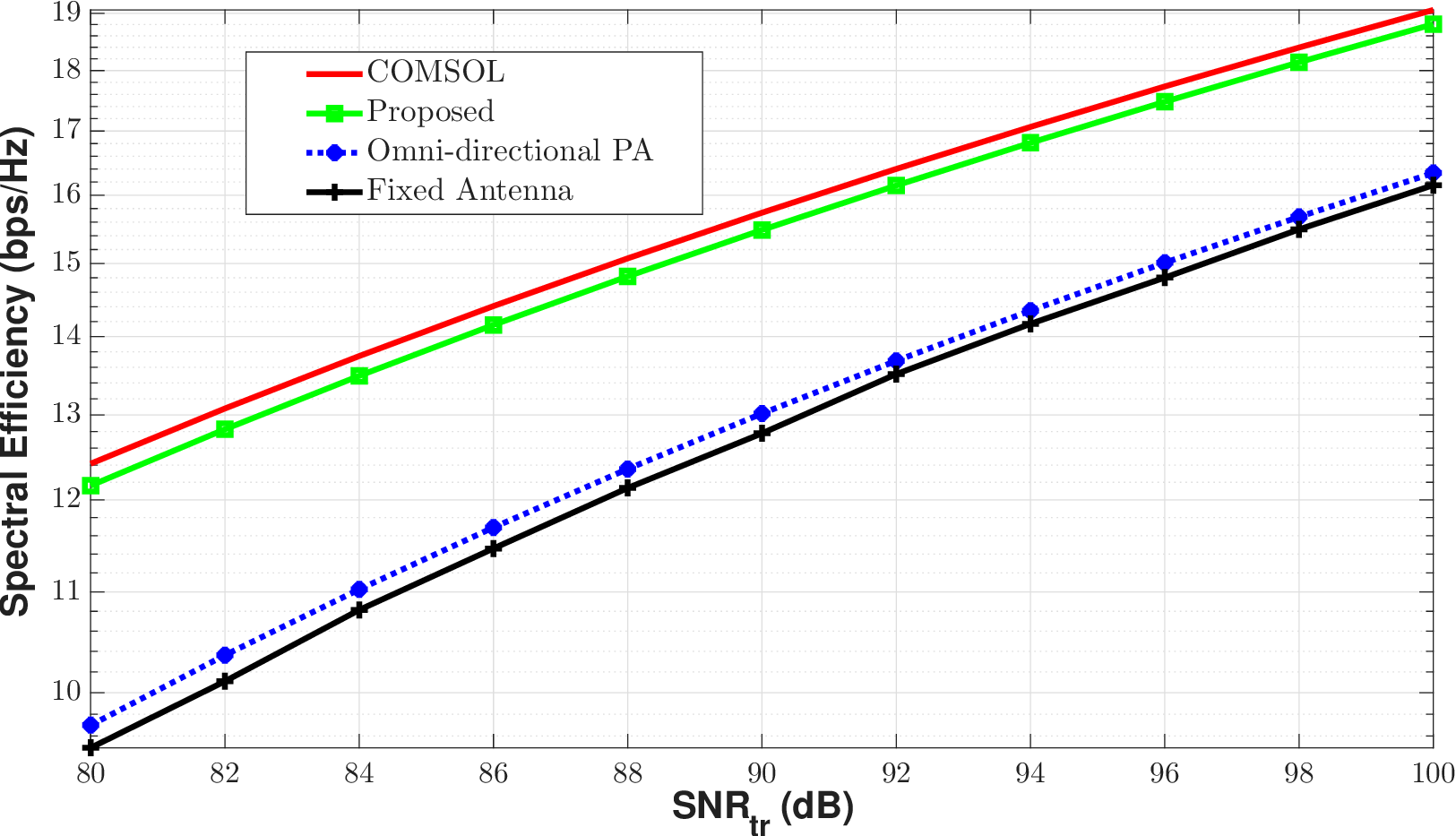}
    \caption{Spectral efficiency versus transmit SNR for $L_s =2\lambda$ and $V_q = 1.5$. }
    \label{fig:InfoRate}
    \vspace{-1mm}
\end{figure}

\section{Conclusions}
In this paper, we have proposed a closed-form analytical model for the radiation pattern of a 2D PASS based on coupled-mode theory and a two-dimensional radiation integral. The model predicts the beam direction, beamwidth, and coupling length directly from the waveguide and PA parameters without empirical fitting. Validation against full-wave COMSOL simulations confirmed close agreement across a range of PA dimensions. Furthermore, the results demonstrated that neglecting the directional nature of the PA radiation pattern leads to suboptimal PA placement and significant performance loss in terms of spectral efficiency. Owing to its closed-form structure, the proposed model can be readily integrated into existing PASS system-level frameworks to replace the omni-directional assumption with a physically grounded directional gain function. Extension to a full 3D radiation-pattern model is left for future work.

\vspace{-2mm}
\appendices
\section{Proof of Theorem \ref{theorem:kappa}} \label{app:coupling}
The coupling coefficient of a dielectric slab waveguide-based PASS can be expressed as~\cite{Book:okamoto2021fundamentalsOpticalWG}
\begin{equation}
    \kappa_{ms}  = \frac{\omega \epsilon_{0} \left(n_{1}^{2} - n_{0}^{2}\right)\int_{-\frac{W_m}{2}}^{\frac{W_m}{2}} E_{m}^{*}\left(y\right)E_{s}\left(y\right) \mathrm{d}y}{\frac{2\beta_{mx}}{\omega \mu_{0}} \int_{-\infty}^{\infty} | E_{m}\left(y\right)|^{2} \mathrm{d}y}.
    \label{eq:couplingCoefficient}
\end{equation}
Here, $E_{m}(y)$ and $E_{s}(y)$ represent the transverse electric field profiles as defined in (\ref{eq:2DFields_MainGuide}) and (\ref{eq:2DFields_Pinch}), respectively, $\omega$ is the angular frequency, and $(\cdot)^*$ denotes complex conjugation. 

The numerator integral in (\ref{eq:couplingCoefficient}) must be evaluated over the main waveguide core, i.e., in the interval $\left(\tfrac{-W_m}{2}, \tfrac{W_m}{2}\right)$, whereas the denominator integral must be taken over both the core and cladding regions. Within the core, the electric field of the main waveguide can be expressed as $E_m(y) = A_0\cos\!\left(\beta_{my} y\right)$. The integral in the denominator of (\ref{eq:couplingCoefficient}) is the same $\forall s$ and is given by
\begin{align}
    N_m & =\int_{-\infty}^{\infty} \left| E_m(y) \right|^2\! \mathrm{d}y \notag \\
    &= A_{0}^{2}\left(\frac{W_m}{2}\! +\! \frac{\sin(\beta_{my} W_m)}{2\beta_{my}} \!+\! \frac{\cos^2(\beta_{my} W_m / 2)}{\sigma_m}\right)\!.
\label{eq:AppendixB_Nm}
\end{align}

We define the integral in the numerator in (\ref{eq:couplingCoefficient}) for the upper PA as
\begin{equation}
    I_{mu} = \int_{-W_m/2}^{W_m/2}E_m^{*}(y) E_u(y) \mathrm{d}y.
\label{eq:AppendixB_Imu_Imd}
\end{equation}

In the region $y\in \left[ \frac{-W_m}{2}, \frac{W_m}{2}\right]$, the evanescent field of the upper PA field profile can be expressed as $E_u(y) = A_{0}\cos\!\left(\frac{\beta_{uy} W_u}{2}\right) e^{-\sigma_u\left(\frac{W_m}{2} - y\right)}$, where a symmetry condition of equal amplitude, i.e., $B_{0u}= A_0$~\cite{Book:okamoto2021fundamentalsOpticalWG} is exploited. Substituting $E_m(y)$ and $E_u(y)$ into (\ref{eq:AppendixB_Imu_Imd}), $I_{mu}$ can be expressed as
\begin{align}
    I_{mu} & = A_{0}^{2}\cos\!\left(\frac{\beta_{uy} W_u}{2}\right)e^{-\sigma_u\left(\frac{W_m}{2}\right)}\!\!
\int_{-\tfrac{W_m}{2}}^{\tfrac{W_m}{2}} \!\!\!\cos(\beta_{my} y) e^{\sigma_uy} \mathrm{d}y.
\label{eq:Imu_1}
\end{align}
Analogously, we define $I_{md}$ using the evanescent tail of the lower PA, which for $y > -\frac{W_m}{2}$ is given by $E_d(y) = A_{0}\cos\!\left(\frac{\beta_{dy} W_d}{2}\right) e^{-\sigma_d\left(y + \frac{W_m}{2}\right)}$.
Since we have $\sigma_d = \sigma_{u}= \sigma_{s}$ and $W_d= W_u = W_s$, we have $I_{md} = I_{mu}=I_{ms}$. Using \cite[Eq. 2.662/2]{Book:Gradshteyn2014table}, the expression for $I_{ms}$ is determined, and the final solution for $I_{ms}$ can be written as
\begin{align}
I_{ms} 
&= A_{0}^{2}\frac{\cos\!\left(\frac{\beta_{sy} W_s}{2}\right) e^{-\sigma_s \left(\tfrac{W_m}{2}\right)}}{
\sigma_s^2 + \beta_{my}^2
}\Biggl[2\sigma_s \cos\!\left(\frac{\beta_{my} W_m}{2}\right)
\times
\nonumber \\ &
\sinh\!\left(\frac{\sigma_s W_m}{2}\right)
+
2\beta_{my} \sin\!\left(\frac{\beta_{my} W_m}{2}\right)\cosh\!\left(\frac{\sigma_s W_m}{2}\right)
\Biggr].
\label{eq:AppendixB_Imu3}
\end{align}
Substituting (\ref{eq:AppendixB_Imu3}) and (\ref{eq:AppendixB_Nm}) in (\ref{eq:couplingCoefficient}) and after some algebraic manipulations, the solution for $\kappa_{ms}$ leads to \eqref{eq:kappac}, hence $\kappa_{mu}=\kappa_{md}=\kappa_m$.
\vspace{-3mm}
\section{Proof of Theorem {\ref{theorem:AxBx}}}\label{app:axbx}
Using the local coordinate $\tilde{x} = x - \left(x_p - \tfrac{L_{s}}{2}\right)$, differentiating (\ref{eq:coupledModeMain1}) with respect to $\tilde{x}$, and substituting (\ref{eq:coupledModeMain2}) yields
\begin{equation}
\frac{\mathrm{d}^2 A(\tilde{x})}{\mathrm{d}\tilde{x}^2} = -2\kappa^2 A(\tilde{x}), \label{eq:ApendixI_coupledMode_MG}
\end{equation}
which is a second-order harmonic oscillator with $\varpi = \sqrt{2}\kappa$ and general solution $A(\tilde{x}) = C_1 \cos(\varpi \tilde{x}) + C_2 \sin(\varpi \tilde{x})$~\cite{Book:Pain1970PhysicsofVibrationWaves}. Applying the initial conditions $A(0) = A_0$ and $B(0) = 0$ gives $C_1 = A_0$ and $C_2 = 0$, yielding~(\ref{eq:modeAmplitudes}).
\vspace{-2mm}
\section{Proof of Theorem \ref{theorem:RadPattern}} \label{app:rp}
According to the physical optics approximation~\cite{Book:stutzman2012antennaTheoryAndDesign}, the PA field distribution in \eqref{eq:FieldInPinch} is modeled as a tangential aperture field $\mathbf{E}_{t}(x,y) = \hat{\mathbf{q}}\tilde{E}_{s}(x,y)$, where $\hat{\mathbf{q}}$ denotes the unit vector in the polarization direction. Given that the PA electric field is $y$-polarized, we have $\hat{\mathbf{q}} = \hat{\mathbf{y}}$, and thus the tangential field can be written as $\mathbf{E}_{t}(x,y) = \hat{\mathbf{y}}\tilde{E}_{s}(x,y)$. By the equivalence principle~\cite{Book:stutzman2012antennaTheoryAndDesign}, the equivalent magnetic current density follows as $\mathbf{M}_{\mathrm{eq}} = \mathbf{E}_{t}(x,y) \times \hat{\mathbf{n}}$, where $\hat{\mathbf{n}}$ represents the unit vector normal to the plane of PAs. Since the PAs lie in the $x-y$ plane, we have $\hat{\mathbf{n}}=\hat{\mathbf{z}}$ and thus  $\mathbf{M}_{\mathrm{eq}}= \hat{\mathbf{x}}\tilde{E}_{s}(x,y)$. The far-field vector potential is given by~\cite{Book:stutzman2012antennaTheoryAndDesign}
	\begin{equation} \small
		\mathbf{F} \propto \iint\limits_S \mathbf{M}_{\mathrm{eq}}(\mathbf{r}')e^{j\beta_{0}\hat{\mathbf{r}}\cdot\mathbf{r}'}\mathrm{d}S,
		\label{eq:FT}
	\end{equation}
where $S$ is the PA surface in the $x$–$y$ plane, which implies $\mathbf{\hat{r}} = \cos(\phi)\hat{\mathbf{x}} + \sin(\phi)\hat{\mathbf{y}}$, $\mathbf{r'} = x'\hat{\mathbf{x}} + y'\hat{\mathbf{y}}$, and the corresponding differential surface element is $\mathrm{d}S = \mathrm{d}x'\,\mathrm{d}y'$. Since $\mathbf{M}_{\mathrm{eq}}(\mathbf{r}') = \hat{\mathbf{x}}B(x')E_{s}\left(y'\right)e^{-j\beta_{sx}x'}$ is only in $x$-direction, upon substitution, (\ref{eq:FT}) reduces to 
     \begin{equation}  \small 
        \mathbf{F} \propto \hat{\mathbf{x}} \iint\limits_S B(x')E_{s}\left(y'\right)e^{-j\beta_{sx}x'}e^{j\left(\beta_{0}\cos\phi x' +  \beta_{0}\sin\phi y'\right)}\mathrm{d}x'\mathrm{d}y',
        \label{eq:VecPotential}
    \end{equation}
where $\hat{\mathbf{x}} = \cos\phi\,\hat{\mathbf{r}} - \sin\phi\,\hat{\boldsymbol{\phi}}$~\cite[C.1.2]{Book:stutzman2012antennaTheoryAndDesign}. The trigonometric prefactors $\cos\phi$ and $-\sin\phi$ that multiply the radiation integral in~(\ref{eq:VecPotential}) are known as projection (or element) factors, whereas the radiation integral itself is known as the pattern factor~\cite{Book:stutzman2012antennaTheoryAndDesign}. The far-field observation point in the azimuth plane is specified as $(r,\phi)$ with unit vectors $\hat{\mathbf{r}}$ and $\hat{\boldsymbol{\phi}}$. The radiated field is computed perpendicular to the direction of $\mathbf{\hat{r}}$~\cite{Book:stutzman2012antennaTheoryAndDesign}, which means the radial components will be zero at the observation point of the far-field azimuth plane. Only the $\hat{\boldsymbol{\phi}}$-component contributes to the far-field, which is tangential to the surface of the azimuth plane~\cite[Sec. 9.1]{Book:stutzman2012antennaTheoryAndDesign}, resulting in the projection factor \(P(\phi) = \hat{\boldsymbol{\phi}}\cdot\hat{\mathbf{x}} = -\sin\phi\). Since the 2D radiation integral in (\ref{eq:VecPotential}) is separable, it can be expressed as a product of independent 1D integrals over each PA variable~\cite[Sec. 9.2.2]{Book:stutzman2012antennaTheoryAndDesign}. Consequently, after substituting $E_{s} = \cos\!\left(\beta_{sy} (y' - y_{s})\right)$, the integral in (\ref{eq:VecPotential}) reduces to
	\begin{align} \small
		&\overline{F}(\phi) = \int_{x_{p}-L_s/2}^{x_{p} + L_s/2}
		B(x')e^{-j\beta_{sx} x'} e^{j k_x x'} \mathrm{d}x'
		 \notag\\ 
		& \times \int_{W_m}^{W_{m} + W_{s}}
		\cos\!\left(\beta_{sy} (y' - y_{s})\right) e^{j k_y y'}\mathrm{d}y'
		= F_x(\phi) F_y(\phi). 
	\end{align}
The total far-field radiation pattern, incorporating the projection factor $P(\phi)$ and the pattern factor $\overline{F}(\phi)$, is given in \eqref{eq:RadPattern}.
 \vspace{-4mm}    
\bibliographystyle{IEEEtran}
\bibliography{references.bib}
\balance
\end{document}